\documentclass[]{article}
\usepackage[letterpaper,margin=1in]{geometry}

\usepackage[T1]{fontenc}
\usepackage[utf8]{inputenc}
\usepackage{xspace}
\usepackage{placeins}
\usepackage{authblk}
\usepackage{graphicx}
\usepackage{hyperref}

\usepackage{amsmath}
\usepackage{amsthm}
\usepackage{amssymb}

\newtheorem{theorem}{Theorem}
\newtheorem{lemma}[theorem]{Lemma}
\newtheorem{remark}[theorem]{Remark}
\newtheorem{corollary}[theorem]{Corollary}

\usepackage{cleveref}

\newcommand{\ER}{\ensuremath{\exists\mathbb{R}}\xspace}

\newcommand{\R}{\mathbb{R}}

\crefname{theorem}{Theorem}{Theorems}
\crefname{figure}{Figure}{Figures}

\title{On the Complexity of Recognizing Nerves of Convex Sets}

\author[1]{Patrick Schnider}
\author[1]{Simon Weber}
\affil[1]{Department of Computer Science, ETH Zurich}

\date{}

\begin{document}

\maketitle

\begin{abstract}
 We study the problem of recognizing whether a given abstract simplicial complex $K$ is the $k$-skeleton of the nerve of $j$-dimensional convex sets in $\R^d$. We denote this problem by $R(k,j,d)$. As a main contribution, we unify the results of many previous works under this framework and show that many of these works in fact imply stronger results than explicitly stated. This allows us to settle the complexity status of $R(1,j,d)$, which is equivalent to the problem of recognizing intersection graphs of $j$-dimensional convex sets in $\R^d$, for any $j$ and $d$. Furthermore, we point out some trivial cases of $R(k,j,d)$, and demonstrate that $R(k,j,d)$ is \ER-complete for $j\in\{d-1,d\}$ and $k\geq d$.
\end{abstract}

\section{Introduction}
Let $G=(V,E)$ be a graph. We say that $G$ is an \emph{intersection graph} of convex sets in $\mathbb{R}^d$ if there is a family $\mathcal{F}$ of convex sets in $\mathbb{R}^d$ and a bijection $V\rightarrow \mathcal{F}$ mapping each vertex $v_i$ to a set $s_i$ with the property that the sets $s_i$ and $s_j$ intersect if and only if the corresponding vertices $v_i$ and $v_j$ are connected in $G$, that is, $\{v_i,v_j\}\in E$. Such graphs are instances of \emph{geometric intersection graphs}, whose study is a core theme of discrete and computational geometry. Historically, intersection graphs have mainly been considered for convex sets in~$\mathbb{R}^1$, in which case they are called \emph{interval graphs}, or for convex sets or segments in the plane.

A fundamental computational question for geometric intersection graphs is the \emph{recognition problem} defined as follows: given a graph $G$, and some (infinite) collection of geometric objects $C$, decide whether $G$ is an intersection graph of objects of $C$. While the recognition problem for interval graphs can be solved in linear time \cite{Interval}, the recognition of segment intersection graphs in the plane is significantly harder. In fact, Matou\v{s}ek and Kratochv\'{i}l have shown that this problem is complete for the complexity class $\exists\mathbb{R}$ \cite{Kratochvil}. Their proof was later simplified by Schaefer \cite{Schaefer}, see also the streamlined presentation by Matou\v{s}ek \cite{Matousek}.

The complexity class $\exists\mathbb{R}$ was introduced by Schaefer and \v{S}tefankovi\v{c} \cite{ETR}. It can be thought of as an analogue of NP over the reals. More formally, the class is defined via a canonical problem called ETR, short for \emph{Existential Theory of the Reals}. The problem ETR is a decision problem whose input consists of an integer $n$ and a sentence of the form
\[\exists X_1,\ldots,X_n\in\mathbb{R}:\varphi(X_1,\ldots,X_n),\]
where $\varphi$ is a quantifier-free formula consisting only of polynomial equations and inequalities connected by logical connectives. The decision problem is to decide whether there exists an assignment of real values to the variables $X_1,\ldots,X_n$ such that the formula $\varphi$ is true.

It is known that $\mathsf{NP}\subseteq\exists\mathbb{R}\subseteq \mathsf{PSPACE}$, where both inclusions are conjectured to be strict. Many problems in computational geometry have been shown to be $\exists\mathbb{R}$-complete, such as the realizability of abstract order types \cite{Mnev}, the art gallery problem \cite{ArtGallery}, the computation of rectilinear crossing numbers \cite{Bienstock}, geometric embeddings of simplicial complexes \cite{Embeddings}, and the recognition of several types of geometric intersection graphs \cite{Cardinal,Evans,Kang,McDiarmid}.

In this work, we extend the recognition problem of intersection graphs of convex sets to the recognition problem of skeletons of nerves of convex sets. Let us introduce the relevant notions. An (abstract) \emph{simplicial complex} on a finite ground set $V$ is a family of subsets of~$V$, called \emph{faces}, that is closed under taking subsets. The \emph{dimension} of a face is the number of its elements minus one. The dimension of a simplicial is the maximum dimension of any of its faces. In particular, a 1-dimensional simplicial complex is just a graph.
The \emph{$k$-skeleton} of a simplicial complex $K$ is the subcomplex of all faces of dimension at most $k$.
Let $\mathcal{F}$ be a family of convex sets in $\mathbb{R}^d$. The \emph{nerve} of $\mathcal{F}$, denoted by $N(\mathcal{F})$ is the simplicial complex with ground set $\mathcal{F}$ where $\{F_1,\ldots,F_m\}\subset\mathcal{F}$ is a face whenever $F_1\cap\ldots\cap F_m\neq\emptyset$.
In other words, the intersection graph of a family $\mathcal{F}$ of convex sets is the $1$-skeleton of the nerve~$N(\mathcal{F})$. Consider now the following decision problem, which we denote by $R(k,j,d)$: given a simplicial complex $K$ by its maximal faces,
decide whether there exists a family $\mathcal{F}$ of $j$-dimensional convex sets in $\mathbb{R}^d$ such that $K$ is the $k$-skeleton of $N(\mathcal{F})$.

In some cases, the $k$-skeleton of a nerve of convex sets uniquely determines the entire nerve: recall \emph{Helly's theorem}~\cite{Helly} which states that for a finite family $\mathcal{F}$ of convex sets, if every $d+1$ of its members have a common intersection, then all sets in $\mathcal{F}$ have a common intersection. Phrased in the language of nerves, this says that if the $d$-skeleton of the nerve $N(\mathcal{F})$ is complete, then $N(\mathcal{F})$ is an $|\mathcal{F}|$-simplex. In other words, we can retrieve the nerve of a family of convex sets in $\mathbb{R}^d$ from its $d$-skeleton by filling in higher-dimensional faces whenever all of their $d$-dimensional faces are present. 

\begin{remark}
The following Helly-type theorem 
implies the analogous statement that a nerve of $j$-dimensional convex sets can be retrieved from its $(j+1)$-skeleton.
\end{remark}

\begin{theorem}\label{thm:Helly-type}
Let $\mathcal{F}$ be a finite family of $j$-dimensional convex sets in $\mathbb{R}^d$. Assume that any $j+2$ or fewer members of $\mathcal{F}$ have a common intersection. Then all sets in $\mathcal{F}$ have a common intersection.
\end{theorem}

This result is likely known, however we could not find a reference for it, so we include a short proof. The proof requires some algebraic topology, in particular the notion of \emph{homology}. For background on this, we refer to the many textbooks on algebraic topology, for instance the excellent work by Hatcher \cite{Hatcher}. For readers not familiar with this concept, the idea of the proof can still be seen by the intuitive notion that $H_k(X)=0$ means that the space $X$ has no holes of dimension $k$.

\begin{proof}
We want to show that the nerve $N(\mathcal{F})$ is an $|\mathcal{F}|$-simplex. Consider a subfamily $\mathcal{F}'\subseteq\mathcal{F}$ and its induced sub-nerve $N(\mathcal{F}')$. By the nerve theorem (see e.g.\ \cite{Hatcher}, Corollary 4G.3), the sub-nerve $N(\mathcal{F}')$ is homotopy-equivalent to the union $\bigcup\mathcal{F}'$ of the sets in $\mathcal{F}'$, implying that the two objects have isomorphic homology groups. As $\bigcup\mathcal{F}'$ has dimension at most $j$, and $\mathcal{F}$ (and thus also $\mathcal{F}'$) is finite, we have that $H_k(\bigcup\mathcal{F}')=0$ for all $k\geq j+1$. Thus $H_k(N(\mathcal{F}'))=0$ for all $k\geq j+1$ and all $\mathcal{F}'\subseteq\mathcal{F}$. On the other hand, the assumption that any $j+2$ or fewer sets have a common intersection implies that the $(j+1)$-skeleton of $N(\mathcal{F})$ is complete and thus $H_k(N(\mathcal{F}'))=0$ for all $1\leq k\leq j$ and \emph{all} subfamilies $\mathcal{F}'\subseteq\mathcal{F}$. Thus, $N(\mathcal{F})$ must be a simplex.
\end{proof}

\section{Containment results}

We start by showing that all considered problems are in the complexity class $\exists\mathbb{R}$.

\begin{theorem}
For all $k,j$ and $d$, we have $R(k,j,d)\in\exists\mathbb{R}$.
\end{theorem}
\begin{proof}
    Similarly to $\mathsf{NP}$, containment in \ER can be proven by providing a certificate consisting of a polynomial number of real values, and a verification algorithm running on the real RAM computation model which verifies these certificates~\cite{Erickson}. As a certificate, we use the coordinates of some point in $\R^d$ for each maximal face of the input complex $K$. These points then describe a family $\mathcal{F}$ of convex sets: Each set $F$ is the convex hull of all points representing maximal faces $S$ of $K$ such that $F\in S$.
    
    Note that if $K$ is the $k$-skeleton of $N(\mathcal{F})$ for some family $\mathcal{F}$ of $j$-dimensional convex sets in $\R^d$, such a certificate must exist: The points can be placed in the maximal intersections of $\mathcal{F}$, and shrinking each set to the convex hull of these points cannot change $N(\mathcal{F})$.

    Such a certificate can be verified by testing that each set $F$ is $j$-dimensional (e.g., using linear programming), and by testing that the $k$-skeleton of $N(\mathcal{F})$ is indeed $K$. The latter can be achieved in polynomial time by computing the intersection of each subfamily $\mathcal{F}'\subseteq \mathcal{F}$ of at most $\min(k+1,d+1)$ sets. If $k\leq d$, this determines the $k$-skeleton of $N(\mathcal{F})$. If $k>d$, the $k$-skeleton of $N(\mathcal{F})$ is determined by the $d$-skeleton of $N(\mathcal{F})$ by Helly's theorem~\cite{Helly}.
\end{proof}

\begin{lemma}
    $R(k,1,1)$ is in $\mathsf{P}$ for any $k\geq 1$.
\end{lemma}
\begin{proof}
    $R(1,1,1)$ is equivalent to recognizing interval graphs, and can thus be solved in polynomial time (see~\cite{Interval}). Since we are considering a family $\mathcal{F}$ of intervals in $\R^1$, the $1$-skeleton of $N(\mathcal{F})$ uniquely determines $N(\mathcal{F})$. By Helly's theorem, $N(\mathcal{F})$ must be the clique complex of its $1$-skeleton. Thus, $R(k,1,1)$ can be solved as follows: Build the graph $G$ given by the $1$-skeleton of the input complex $K$. Test the following four properties: (i) $G$ is an interval graph, (ii) $K$ is at most $k$-dimensional, (iii) every maximal face of $K$ is a clique of~$G$, and (iv) every clique of size $<k$ in $G$ is contained in some maximal face of $K$. Return yes if the answer to all these tests is yes, otherwise return no. All tests can be performed in polynomial time, thus $R(k,1,1)\in\mathsf{P}$.
\end{proof}

For some constellations of $k,j,d$, any simplicial complex of dimension at most $k$ can be realized as the $k$-skeleton of the nerve of $j$-dimensional convex sets in $\R^d$. In this case we say that the problem $R(k,j,d)$ is \emph{trivial}. Evans et al. prove triviality for $R(1,2,3)$:
\begin{lemma}[\cite{Evans}]
    $R(1,2,3)$ is trivial.
\end{lemma}

Furthermore, we can show that if the dimensions $j$ and $d$ get large enough compared to~$k$, the problem also becomes trivial.
\begin{lemma}
    $R(k,2k+1,2k+1)$ is trivial.
\end{lemma}
\begin{proof}
 Wegner has shown that every $k$-dimensional simplicial complex is the nerve of convex sets in $\mathbb{R}^{2k+1}$ \cite{Wegner}. In particular, it is also the $k$-skeleton of a nerve.
\end{proof}

Finally, we prove the following lifting result.
\begin{lemma}
    If $R(k,j,d)$ is trivial, $R(k,j',d')$ is trivial for all $d'\geq d$ and $j\leq j'\leq d'$.
\end{lemma}
\begin{proof}
    We prove that both $j$ and $d$ can be increased by one without destroying triviality, from which the lemma follows.
    
    Any simplicial complex that can be realized in dimension $d$ can also be realized in a $d$-dimensional subspace of $\R^{d+1}$, thus increasing $d$ by one preserves triviality.

    To see that $j$ can be increased, consider a realization of a simplicial complex as the $k$-skeleton of the nerve of a family $\mathcal{F}$ of $j$-dimensional convex sets in $\R^d$. Now, consider any two subfamilies $\mathcal{F}_1,\mathcal{F}_2$ of $\mathcal{F}$, such that $\left(\bigcap_{F\in \mathcal{F}_1}F\right) \cap \left(\bigcap_{F\in \mathcal{F}_2}F\right)=\emptyset$. The two intersections $\bigcap_{F\in \mathcal{F}_1}F$ and $\bigcap_{F\in \mathcal{F}_2}F$ must have some distance $\epsilon$. Consider $\epsilon_{min}$, the minimum of all such $\epsilon$ over all pairs of subfamilies $\mathcal{F}_1,\mathcal{F}_2$. We extrude every object in $\mathcal{F}$ in some direction not yet spanned by the object by some $\epsilon'$ small enough that no intersection $\bigcap_{F\in \mathcal{F}'}F$ for $\mathcal{F}'\subseteq \mathcal{F}$ grows by more than $\epsilon_{min}/3$. This process can not introduce any additional intersections, and thus the nerve of this family of $j+1$-dimensional sets is the same as the nerve of $\mathcal{F}$. We conclude that triviality of $R(k,j,d)$ for $j<d$ implies triviality of $R(k,j+1,d)$.    
\end{proof}

\section{ Existing \texorpdfstring{$\exists\mathbb{R}$}{ER}-Hardness Results}

\begin{lemma}\label[lemma]{lem:k1d}
$R(k,1,d)$ is \ER-hard for $k\geq 1$ and $d\geq 2$.
\end{lemma}
\begin{proof}
    For $k=1$ and $d=2$, this is equivalent to recognizing segment intersection graphs in the plane, which Schaefer \cite{Schaefer} proved to be \ER-hard by reduction from stretchability. Evans et al. \cite{Evans} generalize Schaefer's proof for intersection graphs of segments in $\R^3$ ($k=1$ and $d=3$). Their proof works by arguing that all segments of their constructed graph must be coplanar. 
    Since the argument implies coplanarity no matter the dimension of the ambient space, the proof also implies \ER-hardness for $k=1$ and $d>3$. Furthermore, for any ``yes''-instance of stretchability, the constructed graph can be drawn using segments with no triple intersections. Thus, the proof implies \ER-hardness for $R(k,1,d)$ for $k>1$, as well.
\end{proof}

Schaefer~\cite{Schaefer} furthermore proved that $R(1,2,2)$ is \ER-hard. In the proof of this result, again no triple intersections occur in the representations of ``yes''-instances. Thus the same proof applies to the following lemma.
\begin{lemma}\label[lemma]{lem:k22}
    $R(k,2,2)$ is \ER-hard for any $k\geq 1$.
\end{lemma}

This solves the complexity status of $R(1,j,d)$ for all $j$ and $d$. We summarize these results in the following corollary.
\begin{corollary}
For $k=1$, $R(k,j,d)$ is \begin{itemize}
    \item in $\mathsf{P}$, if $j=d=1$.
    \item \ER-complete, if $j=1$ and $d>2$, or if $j=d=2$.
    \item trivial in all other cases.
\end{itemize}
\end{corollary}

\section{Lifting to Higher Dimensions}

We can extend a lifting result due to Tancer \cite{Tancer} to our setting. For this, the \emph{suspension} of a simplicial complex $K$ with ground set $V$ and face family $F$ is the simplicial complex $S(K)$ with ground set $V\cup\{a,b\}$ and faces $F\cup\{f\cup\{a\}\mid f\in F\}\cup\{f\cup\{b\}\mid f\in F\}$.

\begin{lemma}
Let $K$ be a simplicial complex and let $j\geq d-1$. Then $K$ is a nerve of $j$-dimensional convex sets in $\mathbb{R}^d$ if and only if $S(K)$ is a nerve of $(j+1)$-dimensional convex sets in $\mathbb{R}^{d+1}$.
\end{lemma}

\begin{proof}
We first show that if $K$ is a nerve of convex sets in $\mathbb{R}^d$ then $S(K)$ is a nerve of convex sets in $\mathbb{R}^{d+1}$. For this, let $\mathcal{F}$ be a family of sets in $\mathbb{R}^d$ whose nerve is $K$ and embed them on the hyperplane $x_{d+1}=0$ in $\mathbb{R}^{d+1}$. For each set $F\in\mathcal{F}$ define $F'$ as the cartesian product of $F$ and the segment defined by $-2\leq x_{d+1}\leq 2$. Adding the hyperplanes $x_{d+1}=-1$ and $x_{d+1}=1$, it is easy to see that the nerve of the resulting set family is $S(K)$.

In the other direction, consider a family $\mathcal{F}'$ of $(j+1)$-dimensional convex sets in $\mathbb{R}^{d+1}$ whose nerve is $S(K)$. Let $A$ and $B$ be the convex sets that correspond to the vertices $a$ and $b$, respectively. As $a$ and $b$ are not connected in $S(K)$, the sets $A$ and $B$ must be disjoint. In particular, they can be separated by a hyperplane $h$. For each other set $F'\in\mathcal{F}'$, consider $F:=F'\cap h$ and let $\mathcal{F}$ be the family of these intersections. Note that $\mathcal{F}$ is a family of $j$-dimensional convex sets in $\mathbb{R}^d$. We claim that the nerve of $\mathcal{F}$ is $K$. Indeed, as $K$ is a subcomplex of $S(K)$, every face of $N(\mathcal{F})$ must be a face of $K$. On the other hand, for every face $f$ of $K$, there are points $p_a$ and $p_b$ in $A$ and $B$, respectively, which lie in the intersection corresponding to faces $f\cup\{a\}$ and $f\cup\{b\}$ of $S(K)$, respectively. The intersection of the segment $p_a p_b$ with $h$ lies in the intersection of the sets corresponding to $f$, showing that every face of $K$ must be a face of $N(\mathcal{F})$.
\end{proof}

Combined with the fact that the $d$-skeleton determines the entire nerve, we get the following reduction.

\begin{corollary}
Let $j\in\{d-1,d\}$. If $R(d,j,d)$ is $\exists\mathbb{R}$-hard, then so is $R(d+1,j+1,d+1)$.
\end{corollary}

Using the \ER-hardness of $R(2,1,2)$ and $R(2,2,2)$ implied by \Cref{lem:k1d,lem:k22}, we thus deduce the following

\begin{theorem}
For any $d\geq 2$ and $k\geq d$, the problems $R(k,d-1,d)$ and $R(k,d,d)$ are $\exists\mathbb{R}$-complete.
\end{theorem}

This strengthens a result of Tancer who has shown that $R(d,d,d)$ is NP-hard \cite{Tancer}.

\section{Conclusion}
We have introduced a generalization of the recognition problem of intersection graphs of convex sets and have seen that several existing results in the literature of intersection graphs imply stronger statements in this setting. In particular, the computational complexities of recognizing intersections graphs of convex sets is completely settled. For small $k,j,d$, the current state of knowledge is summarized in the tables in \Cref{fig:tables}. As can be seen, for many decision problems $R(k,j,d)$, the computational complexity is still open. We conjecture that these cases are either \ER-complete or trivial, determining which of the two remains an interesting open problem. Of course, the analogous problems can be defined for objects other than convex sets, giving rise to many interesting open problems.

\begin{figure}[htbp]
    \centering
    \includegraphics[]{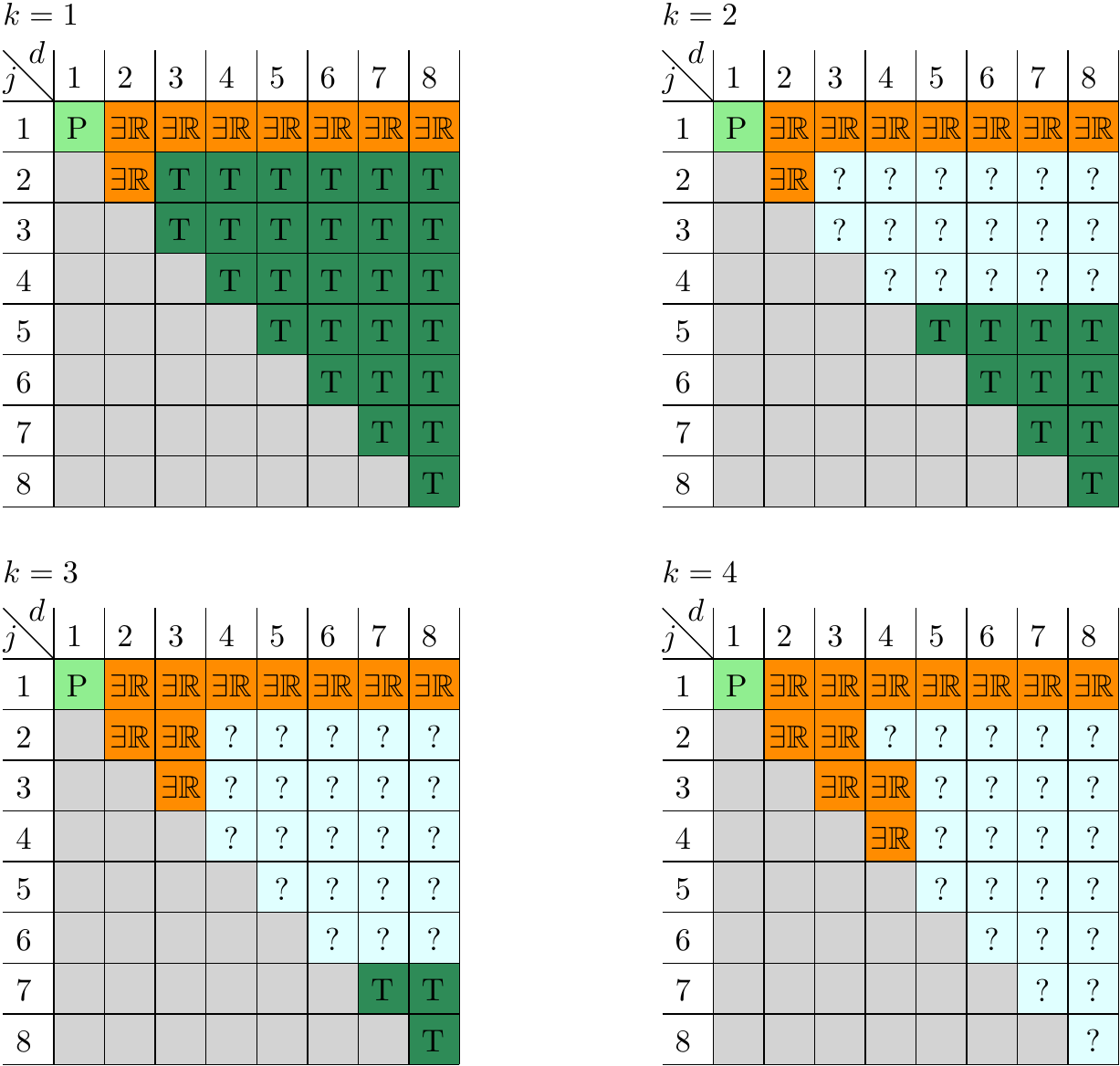}
    \caption{The complexity status of $R(k,j,d)$ for $k\leq 4$ and $d,j\leq 8$. P denotes containment in~$\mathsf{P}$, \ER denotes \ER-completeness, T denotes triviality, and ? indicates open cases.}
    \label{fig:tables}
\end{figure}

\FloatBarrier
\bibliography{literature.bib}

\end{document}